\documentclass{eptcs}
\usepackage{amsmath,amssymb,amsfonts,amsthm}
\usepackage[pdftex]{graphicx,color}
\usepackage{verbatim}
\usepackage{floatflt}
\usepackage{wrapfig}
\usepackage{paralist}

\newcommand{\states}{Q}
\newcommand{\gstates}{S}
\newcommand{\ai}{A^{I}}
\newcommand{\ao}{A^{O}}

\newcommand{\edge}{E}
\newcommand{\trans}{\delta}
\newcommand{\act}{A}

\newcommand{\post}{{\mathsf{post}}}
\newcommand{\asim}{\succeq}
\newcommand{\game}{H}
\newcommand{\weight}{\omega}
\newcommand{\model}{M}
\newcommand{\straa}{\pi^1}
\newcommand{\strab}{\pi^2}
\newcommand{\strai}{\pi^i}
\newcommand{\Straa}{\Pi^1}
\newcommand{\Strab}{\Pi^2}
\newcommand{\Strai}{\Pi^i}

\newcommand{\plays}{\Omega}
\newcommand{\play}{\rho}
\newcommand{\Disc}{Disc}
\newcommand{\Reach}{Reach}
\newcommand{\LimAvg}{LimAvg}
\newcommand{\val}{\nu}
\newcommand{\comments}[1]{}

\newcommand{\mypara}[1]{\noindent{\bf #1}}

\newtheorem{theorem}{Theorem}
\newtheorem{lemma}[theorem]{Lemma}

\newtheorem{remark}[theorem]{Remark}

\title{Interface Simulation Distances\footnote{This research was
    partially supported  
by the European Research Council (ERC) Advanced Investigator Grant
QUAREM, the Austrian Science Fund (FWF) projects S11402-N23 and
S11407-N23 (RiSE), 
the Austrian Science Fund (FWF) Grant No P 23499-N23, 
ERC Start grant (279307: Graph Games) and Microsoft 
faculty fellows award. 
}}
\author{Pavol {\v C}ern\'y \qquad Martin Chmel\'{i}k \qquad Thomas
  A. Henzinger \qquad Arjun Radhakrishna 
\institute{IST Austria}
}
\def\titlerunning{Interface simulation distances}
\def\authorrunning{P. {\v C}ern\'y, M. Chmel\'{i}k, T. A. Henzinger,  
  A. Radhakrishna}

\begin{document}

\maketitle

\begin{abstract}
The classical (boolean) notion of refinement for behavioral 
interfaces of system components is the alternating refinement
preorder.  
In this paper, we define a distance for 
interfaces, called {\em interface simulation distance}. It makes the
alternating refinement preorder quantitative by, intuitively,
tolerating errors (while counting them) in the alternating 
simulation game. 
We show that the interface simulation distance satisfies the triangle
inequality, that the 
distance between two interfaces does not increase under parallel composition
with a third interface, and that the distance between two interfaces can
be bounded from above and below by distances between abstractions of
the two interfaces. We illustrate the framework, and the properties of
the distances under composition of interfaces, with 
two case studies. 
\end{abstract}

\section{Introduction}
The component-based approach is an important design principle in software and
systems engineering. In order to document, specify,
validate, or 
verify components, various formalisms that capture behavioral aspects
of component interfaces have been
proposed~\cite{AH01,Harel87,Jackson00,YS97}. These formalisms capture 
assumptions on the inputs and their order, and guarantees on the
outputs and their order. For closed systems (which do not interact
with the environment via inputs or outputs), a natural notion of
refinement is given by 
the simulation preorder. For open systems, which expect inputs and
provide outputs, the corresponding notion is given by the alternating
simulation preorder~\cite{AHKV98}. Under alternating simulation, an interface A~is
refined by an interface B if, after any given sequence of inputs and
outputs, B accepts all inputs that A~accepts, and B provides only
outputs that A provides. 
The alternating simulation preorder is a boolean notion. Interface A~
either is refined by interface B, or it is 
not. However, there are various reasons for which the alternating
simulation can fail, and one can make quantitative
distinctions between these reasons. For
instance, if B does not accept an input that A accepts (or provides
an output that A does not provide) at every step, then B is more
different from A 
than an interface that makes a mistake once, or at least not
as often as B. 

We propose an
extension of the alternating simulation to the quantitative setting.  
We build on the notion of simulation distances introduced in~\cite{CHR10a}. 
Consider the definition of alternating simulation of an interface A by
an interface B as a two-player game. In this game, Player~1 chooses
moves (transitions), and Player~2 tries to match them. Player~1
chooses input transitions from the interface A and output transitions
from interface B, Player~2 responds by a transition from the other
system. The goal of Player~1 is to prove that the alternating 
simulation does not hold, by driving the 
game into a state from which Player~2 cannot match the chosen move; the goal
of Player~2 is to prove that there exists an alternating simulation,
by playing the game forever. We extend this definition 
to the quantitative case. Informally, we will tolerate errors by
Player~2. However, Player~2 will pay a 
certain price for such errors. 
More precisely, Player~2 is allowed to ``cheat'' by
following a non-existing transition. The price for such transition is
given by an {\em error model}. The error model assigns the transitions from the
original system a weight 0, and assigns the new ``cheating''
transitions a positive weight. 
The goal of Player~1 is then to
maximize the cost of the game, and the goal of 
Player~2 is to minimize it. The cost is given by an objective
function, such as the limit average of transition prices. 
As Player~2 is trying to minimize the value of the game,
she is motivated not to cheat. The value of the game measures how
often Player~2 can be forced to cheat by Player~1. 

\begin{figure}[t]
    \centering
   \resizebox{14cm}{!}{\includegraphics{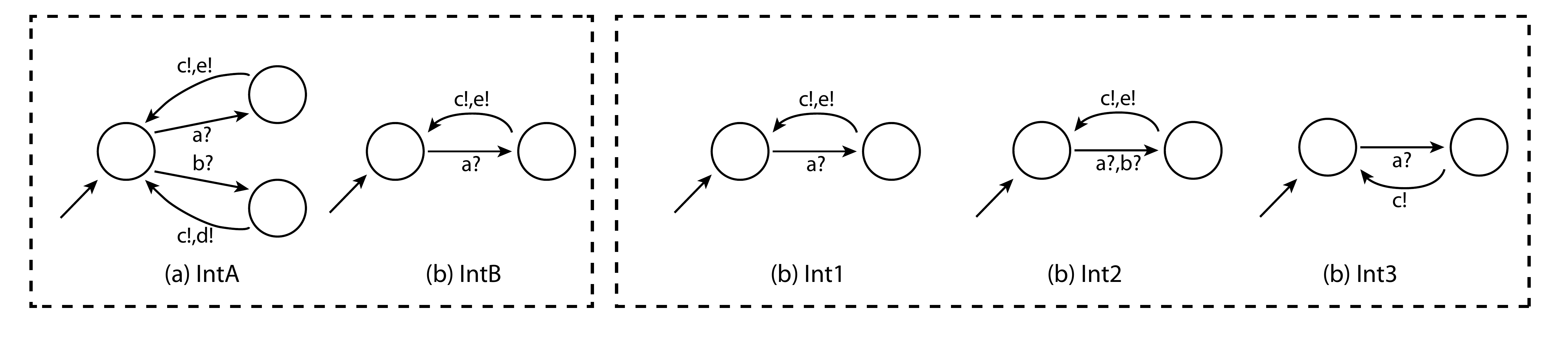}}
    \caption{Example 1}
    \label{fig:introduction}
    \vspace{-0.5cm}
\end{figure}

Consider the example in Figure~\ref{fig:introduction}. The two
interfaces on the left side (IntA and IntB) represent requirements on
a particular component by a designer. The three interfaces (Int1,
Int2, 
and Int3) on the right side are interfaces for different off-the-shelf
components provided by a vendor. We illustrate how interface
simulation distances can be used by the designer to choose a component
whose interface satisfies her requirements most closely. Interface
Int1 is precisely the interface required by IntB, so the distance from
IntB to Int1 will be $0$. However, the distance from IntA to Int1 is
much greater. Informally, this is because Player~1,
choosing a transition of IntA could choose the {\em b?} input. Player~2,
responding by a transition of Int1 has to cheat by playing the {\em
  a?} input. After that, Player~1 could choose the {\em e!} output (as a
transition of Int1), and 
Player~2 (this time choosing a transition from IntA) has to cheat
again. Player~2 thus has to cheat at every step. Interfaces Int2
(resp. Int3) improve on Int1 (with respect to requirement IntA), by
adding inputs (resp. removing outputs). The distance from IntA to Int2 (Int3)
 is exactly half of the distance from IntA to Int1.  
The interfaces Int2 and Int3 have distance 0 to
IntB. Int2 and Int3 satisfy the requirements IntA and IntB better than
the interface Int1.  

The model of behavioral interfaces we consider is a variant of interface
automata~\cite{AH01}. This choice was made for ease of presentation of
the main ideas of the paper. However, the definition of interface
simulation distance can be extended to richer models. 

We establish basic properties of the interface simulation distance. 
First, we show that the triangle inequality holds for the interface
simulation distance. This, together with the fact that reflexivity
holds for this distance as well, shows that it is a {\em directed
  metric}~\cite{dAMRS08}. Second, we give an 
algorithm for calculating the distance. 
The interface simulation distance can be calculated by
solving the value problem in the corresponding game, that is, in
limit-average games or discounted-sum games. 
The values of such games can be computed in pseudo-polynomial
time~\cite{ZP96}. (More precisely, the complexity depends on the
magnitude of the largest weight used in the error model. Thus the
running time is exponential in the size of the input, if the weights
are given in binary.) 

We present composition and abstraction techniques that are useful for
computing and approximating simulation distances between large
systems. These properties suggest that the interface simulation
distance provides an appropriate basis for a quantitative analysis of
interfaces. 
%
The composition of interface automata, which also composes the
assumptions on their environments, was defined in~\cite{AH01}.   
In this paper, we prove that the distance between two interfaces does
not increase under the 
composition with a third interface. The technical challenges in the
proof appear precisely because of the involved definition of composition
of interface automata, and are not present in the simpler setting
closed systems of~\cite{CHR10a}. We also show that the distance between
two interfaces can be over- or under- approximated by distances between
abstractions of the two interfaces. For instance, for
over-approximation, input transitions are abstracted universally, and
output transitions are abstracted existentially.  

We illustrate the interface simulation 
distance, and in particular its behavior under interface composition,
on two case studies. The first concerns a simple message transmission
protocol over an unreliable medium. The second case study models
error correcting codes. 

Summarizing, this paper defines the interface simulation distance for
automata with inputs and outputs, establishes basic properties of
this distance, as well as abstraction and compositionality theorems.   

\mypara{Related work.}
The alternating simulation preorder was defined in~\cite{AHKV98} in order 
to generalize the simulation preorder to input/output systems. The alternating 
simulation can be checked in polynomial time, as is the case for
the ordinary simulation 
relation. Interface automata have been defined in~\cite{AH01} to
facilitate component-based design, and the
theory was developed further, e.g.,
in~\cite{DHJP08,AHS02,CAHS03,larsen_modal_2007}. The 
natural notion of 
refinement for interface automata corresponds to the alternating
simulation preorder. 
Simulation distances have been proposed in~\cite{CHR10a} (the full
version was published recently in~\cite{CHR12}) as a step towards
extending specification formalisms and verification algorithms to a
quantitative setting.  
This paper extends the quantitative notion of simulation distances to the
alternating simulation preorder for interface automata. 

There have been several attempts
to give mathematical  
semantics to reactive processes based on quantitative metrics
rather than boolean preorders~\cite{vanBreugel01,AFS09}. In
particular for probabilistic processes, it is natural to generalize
bisimulation relations to bisimulation metrics~\cite{DGJP04}, and
similar generalizations can 
be pursued if quantities enter through continuous variables, such as
time~\cite{CB02}. 
In contrast, we consider distances between purely discrete
(non-probabilistic, untimed) systems. 
\vspace{-0.6em}
\section{Interface Simulation Distances}
\vspace{-0.3em}

\subsection{Broadcast Interface Automata}

Interface automata were introduced in~\cite{AH01} to
model components of a system communicating through interfaces.
We use a variant of interface automata which we call {\em broadcast
interface automata} (BIA). 

A {\em broadcast interface automaton} $F$ is a tuple $(\states,q^{0},
\ai, \ao, \trans)$ consisting of a finite set of states $\states$, the
initial state $q^{0}$, two disjoint sets $\ai$ and $\ao$ of input and output
actions and a set $\trans \subseteq \states \times \act \times \states$ of
transitions.
We let $\act = \ai \cup \ao$.
Additionally, we require that $F$ is input deterministic, i.e.,
for all $q,q',q'' \in \states$ and all $\sigma_I \in \ai$ if there are
transitions 
$(q,\sigma_I,q')$ and $(q,\sigma_I,q'') \in \trans$, then $q' = q''$.

Given a state $q \in \states$ and an action $\sigma \in \act$ let
$\post(q,\sigma) 
= \{ q' \mid (q,\sigma,q') \in \trans \}$. Similarly given a state $q
\in \states$ let $\ai(q)$ be the input actions enabled at state~$q$
($\ao(q)$ for output actions).
Note that the \emph{BIA} is not required to be input-enabled, hence there
may be states $q$ where $\ai(q) \not = \ai$. 

An example of a \emph{BIA} can be seen on
Figure~\ref{fig:introduction}. The actions terminated by the~$?(!)$
symbol are input (output) actions, respectively. The \emph{BIA}
\emph{IntA} can input $a?$ or $b?$. Depending on the input it can
output $c!$ or $e!$ ($c!$ or $d!$, respectively), and this repeats
forever. 

There are two differences between standard interface automata and BIAs . First, the
communication paradigm in interface automata is pairwise, i.e., an output from a component can
serve as the input to only one other component. 
However, in BIAs the communication model is {\em broadcast}, i.e., an
output from a component can serve as input for multiple different
components. Second, standard interface automata have hidden (internal)
actions, which are omitted from the definition of BIAs. These
modifications were introduced in order to simplify the presentation of
the interface simulation distance, and to enable us to clearly express
the principal ideas. The distance can be
defined for richer models of automata with inputs and outputs,
including for standard interface automata. 

\smallskip\noindent{\bf Alternating Simulation.} Given two {BIAs} $F$
and $F'$, a binary relation on states $\asim \subseteq \states_{F}
\times \states_{F'}$ is an alternating simulation by $F$ of $F'$ if $q
\asim q'$ implies:
\begin{compactenum}
\item for all $\sigma_I \in \ai(q)$ and $r \in \post(q,\sigma_I)$,
  there exists a 
state $r' \in \post(q',\sigma_I)$ such that $r \asim r'$;
\item for all $\sigma_O \in \ao(q')$ and $r' \in \post(q',\sigma_O)$,
there exists a state $r \in \post(q,\sigma_O)$ such that $r \asim r'$. 
\end{compactenum}
%
%
A BIA $F'$ {\em refines} a BIA $F$ (written $F \asim F'$) if
\begin{compactenum}
\item $\ai_{F} \subseteq \ai_{F'}$ and $\ao_{F} \supseteq \ao_{F'}$;
\item there exists an alternating simulation $\asim$ by $F$ of $F'$ such
that $q^{0}_{F} \asim q^{0}_{F'}$.
\end{compactenum}

The intuition behind the above definitions is that when $F \asim F'$,
the component $F$ in a system can be replaced with component $F'$
without leading to any erroneous behavior.

Consider the \emph{BIAs} $\emph{IntB}$ and $\emph{Int2}$ in
Figure~\ref{fig:introduction}. Note that \emph{Int2} refines
$\emph{IntB}$, i.e., \emph{IntB} $\asim$ \emph{Int2}. One can
easily observe that the converse is not true.

\smallskip\noindent{\bf Composition of BIAs.}
When composing BIAs it is required for the inputs (outputs) of
the two automata not to mix, i.e., two BIAs $F$ and $G$ are {\em
composable} if $\ai_F \cap \ai_G = \emptyset$ and $\ao_{F} \cap \ao_{G}
= \emptyset$.
For two composable \emph{BIAs} $F$ and $G$ we let $\emph{shared} (F,G) =
A_{F} \cap A_{G}$.

Whenever there is an action $\sigma \in \emph{shared} (F,G)$ the
composed system makes a joint transition and the output action remains
visible. 
Finally, the {\em composition} of two composable
\emph{BIAs} $F = (\states_F,q^{0}_F, \ai_F, \ao_F, \trans_F)$ and $G =
(\states_G,q^{0}_G, \ai_G, \ao_G, \trans_G)$ is a BIA $F \otimes G =
(\states_{F \otimes G},q^{0}_{F \otimes G}, \ai_{F \otimes G}, \ao_{F
\otimes G}, \trans_{F \otimes G})$ where the states of the product
$\states_{F \otimes G}$ are  $\states_{F} \times \states_{G}$, with the
initial state $q^{0}_{F \otimes G} = (q^0_F,q^0_G)$. The product
input(output) alphabet is $\ai_{F \otimes G} = \ai_{F} \cup \ai_{G}
\setminus \emph{shared}(F,G) \;  (\ao_{F \otimes G} = \ao_{F} \cup
\ao_{G})$, respectively. The transition relation $\trans_{F \otimes
  G}$ contains the 
transition $((q,r),\sigma,(q',r'))$ iff
\begin{compactitem}
\item $\sigma \not \in \emph{shared}(F,G)$ and $(q,\sigma,q') \in
  \trans_{F}$ and 
$r = r'$, or
\item $\sigma \not \in \emph{shared}(F,G)$ and $(r,\sigma,r') \in
  \trans_{G}$ and 
$q = q'$, or
\item $\sigma \in \emph{shared}(F,G)$ and $(q,\sigma,q') \in \trans_{F}$ and
$(r,\sigma,r') \in \trans_{G}$.
\end{compactitem}

Given two composable \emph{BIAs} $F$ and $G$, a product state
$(p,q)$ is an \emph{error} state of the product automaton $F \otimes G$
if there exists a shared action $\sigma \in \emph{shared}(F,G)$ such
that $\sigma \in 
\ao_{F}(p)$ and $\sigma \not \in \ai_{G}(q)$ or  $\sigma \in
\ao_{G}(q)$ and $\sigma 
\not \in \ai_{F}(p)$. A state $(p,q)$ of the product automaton is
\emph{compatible} if no error state is reachable
 from the state $(p,q)$ using only output actions.
\begin{wrapfigure}[9]{}{8cm}
   \resizebox{8cm}{!}{\includegraphics{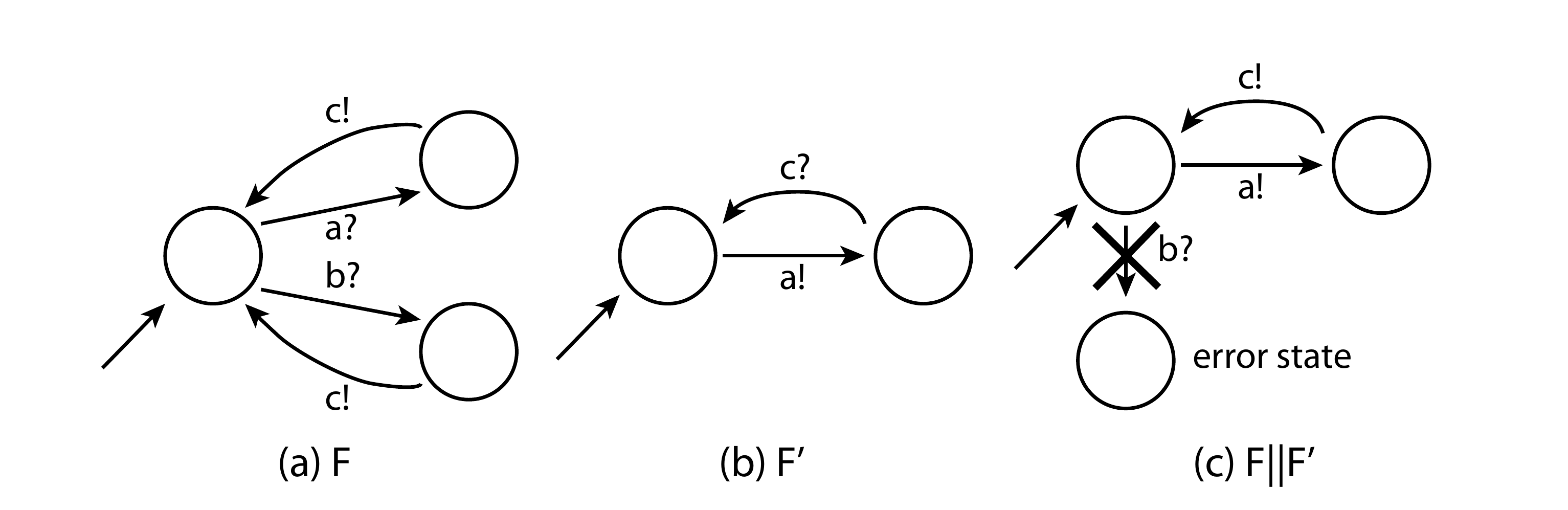}}
    \caption{Composition of \emph{BIAs}}
    \label{fig:composition}
\end{wrapfigure}
A state that is not compatible is
\emph{incompatible}. Two \emph{BIAs} $F$
and $G$ are \emph{compatible} iff the initial state of their product
automaton $F \otimes G$ is 
compatible (denoted by $F \sim G$). The product of two compatible
automata $F$ and $G$ restricted to compatible states is denoted by $F
\parallel G$ and is obtained from $F \otimes G$ by removing input action
transitions that lead from compatible to incompatible states.

A composition of \emph{BIAs} $F$ and $F'$ and the composed interface
$F \parallel F'$ restricted to compatible states can be seen on
Figure~\ref{fig:composition}. Actions $a$ and $c$ become shared
actions in the composition and the composed interface makes a joint
transition on these actions. Note that when constructing the product
$F \otimes F'$ an error state is reachable, therefore the input
transition $(b?)$ leading from a compatible to an incompatible state
is removed from the product $F \parallel F'$. 
\subsection{Graph Games}
\vspace{-0.3em}

In this section, we introduce concepts from the theory of $2$-player
graph games that are necessary for the exposition.
A {\em game graph} is a tuple $\game = (\gstates, \gstates_{1},
\gstates_{2}, \edge, s_\iota)$, where $\gstates$ is a finite set of
states, $\edge \subseteq \gstates \times \gstates$ is a set
of edges, $s_\iota \in \gstates$ is an initial state, and $\gstates_1$
and $\gstates_2$ partition the state space $\gstates$ into Player~1
and Player~2 states respectively.
The game proceeds as follows: First, a token is placed on the initial
state $s_\iota$. Now, whenever the token is on a state $s \in S_i, i \in
\{1,2\}$ Player~$i$ picks a successor $s'$ of $s$ and the token is moved
to the state $s'$, and the process continues infinitely.
The result $\rho = \rho_0\rho_1\ldots$ of an infinite sequence of
visited states is called a {\em play}.  The set of all plays is denoted
by $\plays$.

\noindent{\bf Strategies.}
A {\em strategy for Player~$i$} is a recipe for Player~$i$ to choose the
next transition. Formally, a Player~$i$ strategy $\strai : \gstates^*
\cdot \gstates_i \to \gstates$ is a function
such that for all $w \in \gstates^*$ and $s \in \gstates_i$, we have $(s,
\straa(w \cdot s)) \in \edge$.
We denote by $\Strai$, the set of all Player~$i$ strategies.
The string $w$ is called the {\em history} of
the play and $s$ is the called the {\em last state} of the play.

We define two restricted notions of strategies that are sufficient in
many cases. A strategy is:
\begin{compactitem}
\item {\em Positional} or {\em memoryless} if the chosen successor is
independent of the history, i.e., for all $w \in \gstates^*$,
$\strai(w \cdot s) = \strai(s)$.
\item {\em Finite-memory} if there exists a finite memory set $M$ and an
initial memory state $m_0 \in M$, a memory function $\mu : \gstates^*
\times M \to M$, and a move function $\nu : M \times \gstates_i \to
\gstates$ such that: 
(a) $\mu(\epsilon, m_0) = m_0$ and $\mu(w\cdot s, m_0) = \mu(s, \mu(w,
m_0))$; and
(b) $\strai(w\cdot s) = \nu(\mu(w, m_0),s)$.
Intuitively,
(a)~the state of the memory
is updated based only upon the previous state of the memory and the last
state of the play; and
(b)~the chosen successor depends only on the state of the
memory and the last state of the play.
\end{compactitem}

A play $\rho = \rho_0\rho_1\ldots$ is {\em conformant}
to a Player~$i$ strategy $\pi^i$ if for every $\rho_j \in \gstates_i$,
we have $\pi^i(\rho_0\ldots\rho_j) = \rho_{j+1}$.
Given a game graph $\game$ and strategies $\straa$ and $\strab$ for Player~1
and Player~2 respectively, we get a unique path
$Out_{H}(\straa,\strab)$ that is conformant to both of the
strategies.

\noindent{\bf Objectives.}
A {\em boolean objective} $\Phi \subseteq \plays$ denotes that Player~1
wins if the resultant play $\play$ is in $\Phi$, and that Player~2 wins
otherwise.
A Player~$i$ {\em winning strategy} is one for which all plays
conformant to it are winning for Player~$i$.
We deal with only the {\em reachability} boolean objective. Given a set
of target states $T \subseteq \gstates$ and a play $\play =
\rho_0\rho_1\ldots$, $\play \in \Reach_T$ if and only if
$\exists i : \rho_i \in T$.

A {\em quantitative objective} is a real--valued function $f: \plays
\rightarrow \mathbb{R}$ and the goal of Player~1 is to maximize the
value of the play, whereas the goal of Player~2 is to minimize it. 
We consider the following quantitative objectives:
Given a weight function $\weight : \edge \to \mathbb{R}$, we have
\begin{compactitem}
\item LimAvg$(\rho) = \liminf\limits_{n \rightarrow \infty} \frac{1}{n} \cdot
\sum\limits^{n-1}_{i=0} \weight((\rho_i, \rho_{i+1}))$
\item $\Disc_\lambda(\rho) = \lim\limits_{n \rightarrow \infty}
\sum\limits^{n-1}_{i=0} \lambda^i \cdot \weight((\rho_i, \rho_{i+1}))$
\end{compactitem}

Given a quantitative objective $f$ and a Player~1 strategy $\straa$, the
\emph{value of strategy $\straa$}, denoted by $\val_1(\straa \in \Straa)$ is
$\inf_{\strab \in \Strab} f(Out_H(\straa, \strab))$. Similarly, value
$\val_2(\strab)$ of a Player~2 strategy $\strab$ is $\sup_{\straa \in
\Straa}
f(Out_H(\straa, \strab))$.
The {\em value of the game} is defined as $\sup_{\straa \in \Straa} \val_1(\straa)$
or equivalently, $\sup_{\straa \in \Straa} \inf_{\strab \in \Strab}
\val(Out_H(\straa, \strab))$.
A strategy is \emph{optimal} if its value is equal to the value of the
game.
We conclude this section by stating the memoryless-determinacy theorems
for $\LimAvg$ and $\Disc$ objectives (see e.g.\cite{ZP96}).
\begin{theorem}
For any game graph $\game$ and a weight function $\weight$, we have that
$\sup_{\straa \in \Straa} \inf_{\strab \in \Strab} f(Out_H(\straa,
\strab)) = \inf_{\strab \in \Strab} \sup_{\straa \in \Straa}
f(Out_H(\straa, \strab))$ for $f \in \{ \LimAvg, \Disc \}$. Furthermore,
there exist memoryless optimal strategies for both players.
\end{theorem}
\vspace{-0.7em}
\subsection{Interface Simulation Games}
\vspace{-0.3em}
Simulation like relations can be characterized as the existence of
winning strategies in $2$-player games known as simulation games.
Here, we present the analogue of simulation games for alternating
simulation of BIAs.

\noindent{\bf Alternating Simulation Games.}
Intuitively, given BIAs $F$ and $F'$, Player~1 picks either an
input transition from the $F$ or an output transition from $F'$,
and Player~2 has to match with a corresponding transition with the
same action from $F'$ or $F$, respectively.  We have $F \asim F'$
if and only if Player~2 can keep matching the transitions forever.

Given \emph{BIAs} $F =  (\states_{F},q^{0}_{F}, \ai_{F}, \ao_{F},
\trans_{F})$ and $F' =  (\states_{F'},q^{0}_{F'}, \ai_{F'},
\ao_{F'}, \trans_{F'})$, such that $\ai_{F} \subseteq
\ai_{F'}$ and $\ao_{F} \supseteq \ao_{F'}$, the {\em alternating
simulation game} $\game_{F,F'} = (\gstates,\gstates_{1},\gstates_{2},
\edge,s^0)$ is defined as follows:
\begin{compactitem}
\item The state-space $\gstates = \gstates_{1} \cup \gstates_{2}$, where
$\gstates_{1} = \{(s,\#,s') \mid s \in \states_{F},s' \in
\states_{F'}\} \cup \{s_{\mathsf{err}}\}$ and $
\gstates_{2} = \{(s,\sigma,s') \mid s \in \states_{F}, s' \in
\states_{F'}, \sigma \in \Sigma\}$.
\item The initial state is $s^{0} = (q^{0}_{F}, \#, q^{0}_{F'})$;
\item The Player~1 edges correspond to:
\begin{compactitem}
\item Either input transitions from $F$:
$(s,\#,s') \rightarrow (t,\sigma_I,s') \in E 
\Leftrightarrow s \stackrel{\sigma_I}{\rightarrow} t \in \trans_{F}$; or
\item Output transitions from $F'$:
$(s,\#,s') \rightarrow (s,\sigma_O,t') \in E 
\Leftrightarrow s' \stackrel{\sigma_O}{\rightarrow} t' \in \trans_{F'}$.
\end{compactitem}
\item The Player~2 edges correspond to
\begin{compactitem}
\item Either input transitions from $F'$: $(t,\sigma_I,s') \rightarrow
(t,\#,t')\in E \Leftrightarrow s' \stackrel{\sigma_I}{\rightarrow} t' \in \trans_{F'}$; or
\item Output transitions from $F$:
$(s,\sigma_O,t') \rightarrow (t,\#,t') \in E 
\Leftrightarrow s \stackrel{\sigma_O}{\rightarrow} t \in \trans_{F}$.
\end{compactitem}
\item For all states $s \in \gstates_{2}$ if there is no outgoing edge from
$s$ we make an edge $s \rightarrow s_{\mathsf{err}}$; and for all states $s \in
\gstates_{1}$ if there is no outgoing edge from $s$ we make a selfloop
on $s$.
\end{compactitem}
The objective of Player~1 is to reach the state $s_{\mathsf{err}}$ and
the objective of Player~2 is to avoid reaching $s_{\mathsf{err}}$.

We have the following theorem.
\begin{theorem}
Given BIAs $F$ and $F'$ and the corresponding alternating simulation
game $\game_{F, F'}$, we have that $F \asim F'$ if and only if Player~2
has a winning strategy in $\game_{F, F'}$.
\end{theorem}
\vspace{-0.7em}
\subsection{Quantitative Interface Simulation Games}
\vspace{-0.3em}

We aim to establish a distance function between broadcast interface
automata that expresses how ``compatible" the automata are, even when
the standard boolean notion of refinement is not true. In order to do
that we give more power to Player~2, by allowing him to play actions
that are not originally in the game. However, to avoid free use of such
actions every time Player~2 plays the added action he receives a
penalty. As we do not want Player~2 to play completely arbitrarily we
formalize the allowed ``cheating" by a notion of input (output) error
models.

An {\em input (output) error model} is a function $M: \ai \times \ai \rightarrow
\mathbb{N} \cup \{\bot\}$ resp. $(\ao \times \ao \rightarrow \mathbb{N}
\cup \{\bot\})$. We require that for all $a, b, c \in \ai (\ao)$ that 
$\model(a,a) = 0$ and $\model(a, b) + \model(b, c) \geq \model(a, c)$.
Given a \emph{BIA} $F =  (\states,q^{0}, \ai, \ao,
\trans)$ and an error model $\model$, let the {\em modified system} be
 $F\otimes{\model} = (\states ,q^{0},\ai,
\ao, \trans^{e})$ with a weight function $\weight_{\model}: \trans^{e}
\rightarrow N$, where the terms are defined as follows:
\begin{compactitem}
\item $(s, \sigma_{2}, t) \in \trans^{e} \Leftrightarrow 
((s, \sigma_{1}, t) \in \trans \wedge \model(\sigma_{1},\sigma_{2}) \not
= \bot)$ ;
\item $\weight_{\model}((s, \sigma_{2}, t)) = \min_{(s, \sigma_1,
t) \in \trans} \{\model(\sigma_{1},\sigma_{2})\}$;
\end{compactitem}
Note, that the automata enhanced with input error models 
are not \emph{BIAs} as they may not be input deterministic. However, all the definitions for \emph{BIAs} can be naturally
interpreted on a \emph{BIA} composed with an error model.
Given two \emph{BIAs} $F$ with an
output error model $\model_O$ and $F'$ with an input error model
$\model_{I}$ we construct a game
$\game_{F\otimes{\model_{O}},F'\otimes{\model_{I}}}$ for systems
$F\otimes{\model_O}$ and $F' \otimes \model_I$ similarly as is described for
\emph{BIAs} in the previous subsection. We
measure the ``cheating'' performed by Player~2 as either the
limit-average or the discounted sum of the weights on the transition.
The transitions going out from Player~1 states get weight $0$ with an
exception of the selfloop on $s_{err}$ state that gets the maximal
weight assigned by the error model.
The weight of an edge from a Player~2 state is either (a)~twice the
weight of the corresponding  $F' \otimes M_I$ transition when matching
inputs; or (b)~twice the weight of the corresponding $F \otimes M_O$
transition when matching outputs. The factor $2$ occurs due to
normalization.
Given two BIAs $F$ and $F'$, a quantitative objective $f \in \{$LimAvg,
$\Disc_{\lambda}\}$ and an input (output) error model $\model_I
(\model_O)$, the {\em interface simulation distance} $d^f(F \otimes
\model_I, F' \otimes \model_O)$ is defined to be
the value of game $\game_{F\otimes{M_{O}},F'\otimes{M_{I}}}$.
Consider again the example in Figure~\ref{fig:introduction}, when using
error models that can play input (output) actions interchangeably by
receiving penalty $1$. The distances $d$ among the systems for the
quantitative objective LimAvg are presented in
Table~\ref{tab:distances}.
\begin{wraptable}{r}{4cm}
\centering
\begin{tabular}{c||ccc}
 & Int1 & Int2 & Int3 \\
\hline
IntA & 1 & 1/2 & 1/2 \\
IntB & 0 & 0 & 0\\
\end{tabular}
\caption{Ex. 1}
\label{tab:distances}
\vspace*{-0.5cm}
\end{wraptable}
The result $d($IntA, Int1$) = 1$ is surprising when comparing to
simulation where the distance would be $0$. The high distance is due to
the alternating matching. Player~1 chooses to play input $b?$ in IntA.
Player~2 has no choice but to respond by $b?$ and receiving the first
penalty. Again Player~1 plays the $e!$ output action forcing the second
Player~2 to cheat again. By repeating these transitions Player~1 can
force Player~2 to receive a penalty in every turn and therefore the
distance is 1.
The distance can be improved by adding an $b?$ input action as is shown
in the case of $Int2$, where the distance has decreased to $1/2$. Player~2
can now match every possible input, but fails to react on the $e!$
output action. Player~1 can ensure the value $1/2$ by playing $b?$ an
$e!$ repeatedly.
The second option to improve the distance is to remove some of the
output edges as is shown in $Int3$. Player~2 still cannot match the
input $b?$ but can respond to $c!$ without receiving any penalty. As in
the previous case playing a sequence $b?,c!$ ensures value $1/2$ for
Player~1.

\vspace{-0.5em}
\subsection{Complexity}
\vspace{-0.3em}
By the results presented in \cite{ZP96} the complexity of finding the
value of the game for \emph{LimAvg} objective is in $\mathcal{O}(|V|^3
\cdot |E| \cdot W)$, where $|V|$ is the number of game states, $|E|$ is
the number of edges and $W$ is the maximal non--infinite weight used in the game. In
our case for \emph{BIA} $F = (\states_F,q^{0}_F, \ai_F, \ao_F,
\trans_F)$ and $G  = (\states_G,q^{0}_G, \ai_G, \ao_G, \trans_G)$ and
error model $\model_O, \model_I$, the number of states in the game
$\game_{F \otimes \model_O, G \otimes \model_I}$ is $|\states_F| \cdot
|\states_G| \cdot (|A_F| + |A_G| + 1)+1$ and the number of edges is bounded by $|V|^2$.  
The algorithm for $\Disc_{\lambda}$ given a fixed $\lambda$ is in PTIME by
a variation of an algorithm presented in \cite{ZP96}.

\vspace{-0.7em}
\section{Properties of Interface Simulation Distances}
\vspace{-0.3em}
In this section, we present properties of the interface simulation
distance. The distance satisfies the triangle inequality and does not
increase when \emph{BIAs} are composed with a third
interface. Moreover, the distance can be bounded from above and below
by considering the abstractions of the systems.  

\vspace{-0.5em}
\subsection{Triangle Inequality}
\vspace{-0.3em}
The triangle inequality is the quantitative analogue of the boolean
transitivity property. We show that the interface simulation distance is
a directed metric, i.e., it satisfies the triangle inequality and the
reflexivity property. The proof is similar to the case of pure simulation distances presented
in \cite{CHR10a}.
\begin{theorem}
For $f \in \{$LimAvg, $\Disc_{\lambda}\}$ the interface simulation distance $d^f$ is a directed metric, i.e.:
\begin{compactenum}
\item For all error models $\model_{I}, \model_{O}$ and BIAs $F_1,F_2,F_3$ we have:
$$d^f(F_1 \otimes \model_{O}, F_3 \otimes \model_{I}) \leq d^f(F_1 \otimes
\model_{O}, F_2 \otimes \model_{I}) + d^f(F_2 \otimes \model_{O}, F_3
\otimes \model_{I})$$
\item For all error models $\model_{I}, \model_{O}$ and BIAs $F$ we have $d^f(F
	\otimes \model_{O}, F \otimes \model_{I}) = 0$.
\end{compactenum}
\end{theorem}

\subsection{Composition}

In this part we show that the distance between two \emph{BIAs} 
$F$ and $F'$ does not increase when both are composed with a third
 \emph{BIA} $G$, when using the same error models $\model_{O}$,
$\model_{I}$.

As we want to use the same output error model $\model_{O}$
in $F$ and $F\parallel G$ (similarly $\model_{I}$ in $F'$ and $F'
\parallel G$), we restrict the error models. Assume $\sigma_1 \not =
\sigma_2$, then:
\begin{compactitem}
\item If $\model_{O}(\sigma_{1},\sigma_{2}) \not = \bot$, then
$\sigma_{2} \in \ao_{F \parallel G} \setminus \ao_{G}$.
\item If $\model_{I}(\sigma_{1},\sigma_{2}) \not = \bot$, then
$\sigma_{2} \in \ai_{F' \parallel G} \setminus \ai_{G}$.
\end{compactitem}

\begin{remark}
\label{lem:alphabet}
By the above restriction on the error models, we get $\ao_{F} = \ao_{F
\otimes M_{O}}$ and $\ai_{F'} = \ai_{F' \otimes M_{I}}$. Therefore, we
get that $F \otimes \model_O$ and $F' \otimes \model_I$ are composable
with $G$ if $F$ and $F'$ are composable with $G$.
\end{remark}

The following lemma establishes that extending $F$ with the error model does
not change compatibility with $G$. Note that this would not be the case
if the assumption on the error models was violated.
\begin{lemma}
\label{lem:compatible}
For all \emph{BIAs} $F,G$ and error model $\model_I, \model_O$, if $F$ is compatible with $G$, then $F \otimes \model_O (\model_I)$ is compatible
with $G$.
\end{lemma}
The proof for the output error model $\model_O$ follows easily from the fact that any
sequence of output actions in $(F \otimes \model_O) \otimes G$ can be
replayed in $F \otimes G$ by replacing those actions that are added by
the error model in $F \otimes \model_O$  with the original transitions
from $F$. The case of input error model follows directly from the definition.

First, we establish the following preliminary lemma in anticipation of
the main theorem. 
We need to show that property of incompatibility and of being error
states is preserved even when the BIAs are extended with error models.
\begin{lemma}
\label{lem:incompatibility}
Let $F$, $F'$, and $G$ be BIAs with $\emph{shared}(F',G) \subseteq  \emph{shared}(F,G)$, and $M_{O}, M_{I}$ error
models. Suppose that $(p', q)$ is a state in $(F' \otimes G) \otimes M_{I}$
and $p \asim p'$ for some alternating simulation relation $\asim \; \subseteq \states_{F}
\times \states_{F'}$ between $F \otimes \model_O$ and $F' \otimes
\model_I$. Then,
\begin{compactenum}
\item $(p',q)$ is an error state, then $(p, q)$ is an
error state;
\item $(p',q)$ is an incompatible state, then $(p, q)$ is an
incompatible state.
\end{compactenum}
\end{lemma}
\begin{proof}
\begin{compactenum}
\item 
From the definition of an error state, it follows that there exists an action $a \in
\emph{shared}(F',G) \subseteq \emph{shared}(F,G)$ such that either,
\begin{compactitem}
\item $a \in \ao_{F'}(p')$ and $a \not \in \ai_{G}(q)$, or
\item $a \in \ao_{G}(q)$ and $a \not \in \ai_{F'}(p')$.
\end{compactitem}
In the former case as $p \asim p'$, we have $a \in \ao_{F}(p)$, hence
$(p,q)$ is an error state. In the latter case from $a \not \in
\ai_{F'}(p')$ and $p \asim p'$ follows that $a \not \in \ai_{F}(p)$ and
again $(p,q)$ is an error state.

\item If $(p',q)$ is an incompatible state in $(F' \otimes G) \otimes
M_{I}$, it follows that an error state is autonomously reachable from
$(p',q)$ using only output actions.
As $p \asim p'$ the same sequence of actions can be replayed from the
state $(p,q)$: (i)~the actions that change only the $G$ component of the
state are the same, and (ii)~the actions that change the $F'$ component
can be simulated in $F$ as $p \asim p'$.
We have either (a)~that the replayed sequence reaches an error state
before the end; or (b)~the last reached state is an error state. 
The claim (b) follows from the previous part~1. In both cases, we get
the required result.
\end{compactenum} 
\end{proof}

The following lemma states that the broadcast interface automata
enhanced with error models have
the same properties on composition as interface automata. Note that the
restrictions on error models imply that the \emph{BIA} composed with an error model
remains input deterministic on
shared actions.  Due to this fact the proof is a
variation of a similar result for  interface automata presented in \cite{AH05}.

\begin{lemma}
\label{lem:qualitative_case}
Consider three BIAs $F, G$, and $F'$ with input (output) error models
$\model_I (\model_O)$, such that $F \otimes \model_O$ and $G$
are composable and $\emph{shared}(F',G) \subseteq  \emph{shared}(F,G)$.
If $F \otimes \model_O \sim G$ and $F \otimes \model_O \asim F' \otimes \model_I$, then $F' \sim G$.
\end{lemma}

Finally, we can prove the main theorem, showing that composition with
a third interface can only decrease the distance. 
In the game between the composed systems, we construct a Player~2
strategy that (a)~for the first component, use the Player~2 strategy
from the game $H_{F \otimes \model_O, F' \otimes \model_I}$; and (b)~for
the second component, copies the Player~1 transition.

There are two obstacles to this scheme of using the Player~2 strategy in
the first component: (a)~some of the actions become shared actions; and
(b)~some of the states of the composed system may become unreachable due
to their incompatibility.
Using Lemma~\ref{lem:incompatibility} and Lemma~\ref{lem:compatible}, we
will overcome the obstacles.
\begin{theorem}
\label{thm:quantitative_case}
Consider three \emph{BIAs} $F,G$, and $F'$, a quantitative objective
$f\in\{$LimAvg,$\Disc_{\lambda}\}$, and input (output)
error models $\model_{I}$, $\model_{O}$  such that $F$ and $G$ are
composable, compatible, and $\emph{shared}(F',G) \subseteq \emph{shared}(F,G)$.
Then,
$$d^f(F \otimes \model_{O},F' \otimes \model_{I}) \geq d^f((F \parallel G)
\otimes M_{O},(F' \parallel G) \otimes M_{I}).$$
\end{theorem}

\begin{proof}
We split the proof into two cases.

\noindent(a)~Player~2 cannot avoid reaching $s_{err}$ state in the
game $\game_{F\otimes M_O, F' \otimes M_I}$. This is the easier case and
we will not present the details here.

%

\noindent(b)~Player~2 can avoid reaching the $s_{err}$ state in
the game $\game_{F\otimes{M_{O}},F'\otimes{M_{I}}}$.
We get that $F'\otimes M_{I}$ refines $F \otimes M_{O}$. Let $\asim'$ be the maximal alternating simulation relation and
furthermore, let $\strab_*$ be the optimal positional Player~2 strategy
in the game.
By Remark~\ref{lem:alphabet} we get that $F \otimes \model_O$ is
composable with $G$, from Lemma~\ref{lem:compatible} it follows that $F
\otimes \model_O$ is compatible with $G$ and finally by
Lemma~\ref{lem:qualitative_case} we get that the composition $F'
\parallel G$ is not empty.


Using the relation $\asim'$, we define an alternating simulation relation $\asim^{*}$
by $(F \parallel G)\otimes M_{O}$ of $(F' \parallel G)\otimes
M_{I}$ as follows: $$ (p,q) \asim^{*} (r,s) \Leftrightarrow p \asim' r
\: \wedge \: q = s\qquad\qquad\mbox{for $p$ and $r$ states of $F$ and $F'$ and $q$
and $s$ states of $G$} $$

We construct a positional Player~2 strategy $\strab$ in the game
$\game_{(F \parallel G)\otimes{M_{O}},(F' \parallel G)\otimes{M_{I}}}$
based on the strategy $\strab_*$, such that for all Player~1 strategies
$\straa$ the strategy $\strab$ will ensure that
$f(out(\straa,\strab)) \leq d^f(F \otimes \model_{O},F'
\otimes \model_{I})$.

We will match actions in the first component using the strategy
$\strab_*$. Actions from the $G$ component are
going to be copied directly. This will ensure that every reachable
Player~1 state $((p,q),\#,(r,s))$ satisfies $(p,q) \asim^{*} (r,s)$.
We have the following cases based on the kind of transition chosen by
Player~1:
\begin{compactitem}
\item \textbf{Unshared actions from $G$:} If Player~1 chooses the state
$((p, q'), \sigma_I, (r, s))$, we have $\strab(((p, q'), \sigma_I, (r,
s))) = ((p, q'), \#, (r, q'))$. This is possible as $q = s$. Similarly
for a state $((p, q), \sigma_O, (r, s'))$ we define $\strab(((p, q),
\sigma_O, (r,
s'))) = ((p, s'), \#, (r, s'))$ 
\item \textbf{Unshared input action from $F$:} If Player~1 chooses the
state $((p', q), \sigma_I, (r, s))$, we have $\strab(((p', q), \sigma_I,
(r, s))) = ((p', q), \#, (r', s))$ if $\strab_*(p', \sigma_I, r) = (p',
\#, r')$. We have to make sure that the transition $((r, s), \sigma_I,
(r', s))$ is not removed to ensure compatibility. In that case, from
Lemma~\ref{lem:incompatibility} and the fact that $(p, q) \asim* (r,
s)$, we would have that $(p, q)$ is an incompatible state. However, the
transition from compatible to incompatible state in the $(F \parallel G)
\otimes M_O$ component is possible only by a Player~1 transition.
Therefore, we have that Player~2 will not play an incompatible
transition if Player~1 does not play an incompatible transition.
\item \textbf{Unshared output action from $F'$:} This case is similar to
the previous, but simpler as output transitions are not removed to
ensure compatibility.
\item \textbf{Shared output action (input from $G$):} If Player~1 chooses the
state $((p, q), \sigma_O, (r', s'))$, we have $\strab(((p, q), \sigma_O,
(r', s'))) = ((p', s'), \#, (r', s'))$ if $\strab_*(p, \sigma_O, r') = (p',
\#, r')$. 
\item \textbf{Shared output action (output from $G$):} This case is the
trickiest due to the need to simulate inputs in the first component the
``wrong'' way (from $F'$ to $F$).

If Player~1 chooses the state $((p, q), \sigma_O, (r', s'))$, we have
$\strab(((p, q), \sigma_O, (r', s'))) = ((p', s'), \#, (r', s'))$ where
$p'$ is the unique state reachable from $p$ on the action $\sigma_O$.
The existence of this action is argued here.
\begin{itemize}
\item Firstly, due to input determinacy on shared actions, at most one state is reachable
from $p$ on action $\sigma_O$. Furthermore, there can be no transitions
with action $\sigma_O$ added by $\model_I$ as $\sigma_O$ is shared with
$G$.
\item Second, assuming that $(p, q)$ is compatible, we have that at
least one state is reachable from $p$ on action $\sigma_O$. As in the
above cases, we can argue that $(p, q)$ is compatible.
\end{itemize}
In the game $\game_{F \otimes M_O, F' \otimes M_I}$, we can translate
this step as follows. From $(p, \#, r)$,
Player~1 chooses the successor $(p', \sigma_O, r)$ (note that $\sigma_O$
is an input action for $F$ and $F'$); and then, $\strab_*((p', \sigma_O,
r)) = (p', \#, r')$. The justification is as follows: Since, $s_{err}$ is
not visited, $\strab_*$ has to choose a successor with the transition symbol
$\sigma_O$ (which is uniquely $p'$, as above). 

\end{compactitem}

Let $\straa$ be an arbitrary Player~1 strategy. If we consider the play
$\rho = out(\straa,\strab)$, 
(a)~If the first case from the $5$ above occurs, the transition weight
is $0$; and (b)~For any of the other cases, the transition weight is the
same as weights from a play $\rho'$ in $\game_{F \otimes \model_O, F' \otimes
\model_I}$ conformant to $\strab_*$.

Therefore, we have that weights in $\rho$ are weights in $\rho'$,
interspersed with some $0$ weights. Hence, we get
$$d^f((F \parallel G)\otimes \model_{O}, (F' \parallel G) \otimes
\model_{I}) \leq f(\rho') \leq f(\rho) \leq \val(\strab_*) = d^f(F \otimes \model_{O},F' \otimes
\model_{I})$$
proving the required result.
\end{proof}
\subsection{Abstraction}
In the classical boolean case, systems can by analyzed with the help of
sound over- and under-approximations. We present the quantitative
analogue of the soundness theorems for over- and under-abstractions.
The distances between systems is bounded by the distances between their
abstractions.

Given a \emph{BIA} $F = (\states,q^{0}, \ai, \ao, \trans)$ a {\em
$\forall \exists$ abstraction} is a \emph{BIA} $F^{\forall\exists} =
(S,[q^0], \ai, \ao, \trans^{\forall\exists})$, where $S$ are the
equivalence classes of some equivalence relation on $\states$ and
\begin{eqnarray*}
 \trans^{\forall\exists} &= & \{ (s, \sigma_I, s') \mid \sigma_I \in \ai
\text{ and } \forall q \in s, \exists q' \in s': (q, \sigma_I, q') \in
\trans \} \cup \\
 & & \{ (s, \sigma_O, s') \mid \sigma_O \in \ao \text{ and } \exists q
\in s, \exists q' \in s': (q, \sigma_O, q') \in \trans \}
\end{eqnarray*}
Similarly we define the $\exists\forall$ abstraction \emph{BIA} with the
transition relation defined as follows:
\begin{eqnarray*}
 \trans^{\exists\forall} &= & \{ (s, \sigma_I, s') \mid \sigma_I \in \ai
\text{ and } \exists q \in s, \exists q' \in s': (q, \sigma_I, q') \in
\trans \} \cup \\
 & & \{ (s, \sigma_O, s') \mid \sigma_O \in \ao \text{ and } \forall q
\in s, \exists q' \in s': (q, \sigma_O, q') \in \trans \}
\end{eqnarray*}

\begin{theorem}
Let $f$ be one of the objectives in $\{$LimAvg,$\Disc_{\lambda}\}$ and $F, G$ be arbitrary {BIAs} with $\model_O, \model_I$ as error models,
then the following inequalities hold:
$$ d^f(F^{\forall\exists} \otimes \model_O,G^{\exists\forall} \otimes
\model_I) \leq d(F \otimes \model_O,G \otimes \model_I) \leq
d^f(F^{\exists\forall} \otimes \model_O,G^{\forall\exists} \otimes
\model_I) $$
\end{theorem}
\begin{proof}
Let $\strab$ be the optimal positional Player~2 strategy in the game $\game_{F \otimes \model_O,G \otimes \model_I}$ we construct a positional Player~2 strategy $\strab_*$ in $\game_{F^{\forall\exists} \otimes \model_O,G^{\exists\forall} \otimes \model_I}$ that is going to ensure the needed value. When defining the strategy, we need to distinguish between two cases:

\textbf{Input actions} Let the state be $(s_F, \sigma_I, s_G)$ for some $\sigma_I \in \ai$. We pick a state $q_F \in s_F$ and $q_G \in s_G$ such that strategy $\strab$ can ensure the value from the state $(q_F, \sigma_I, q_G)$. Let $\strab$ reach state $(q_F, \#, q'_G)$ by playing action $\sigma_I$. Then $\strab_*$ plays action $\sigma_I$ to a state $(s_F, \#, [q'_G])$.

\textbf{Output actions}
Similarly as in the previous case let the state be $(s_F, \sigma_O, s_G)$ for some $\sigma_O \in \ao$. We pick a state $q_F \in s_F$ and $q_G \in s_G$ such that strategy $\strab$ ensures the value from the state $(q_F, \sigma_O, q_G)$. If the state reached by $\strab$ is $(q'_F,\#, q_G)$ then $\strab_*$ reaches a state $([q'_F],\#,q_G)$.

From every play conformant to $\strab_*$ we can extract a play
conformant to $\strab$ such that their values are equal. This concludes
the first inequality. The proof of the second inequality is similar, but
considers the optimal Player~1 strategy.
\end{proof}
\vspace*{-1.0em}
\section{Case Studies}
\vspace*{-0.5em}

We present two case studies illustrating the interface simulation
distances framework. In the first one, we describe a message transfer
protocol for sending messages over an unreliable medium. This case
study also serves to illustrate the behavior of the distance under
interface composition. The second case study is on error correcting
codes. In both cases, we use the limit average objective. 

\comments{
In this section we show an example of a message transfer protocol and
and a case study including error correcting codes. In the first example
we model a message transfer protocol that sends a message over an
unreliable medium. By analyzing 2 different interfaces and measuring
their distances to the specification we show that the interface that
is working ''more" as prescribed is closer according to the distance
introduced in the previous sections.

In the next part we introduce error correcting codes that differ in the
number of bit flips that can be corrected after a noisy
transmission. When comparing the distances to the specification (it
outputs the input no matter how many bit flips occurred during the
transmission) we show that the more bit flips the code can correct the
closer it is to the specification. In both cases we are considering the
quantitative LimAvg objective.
}

\vspace{-0.7em}
\subsection{Message Transmission Protocol}
\vspace{-0.3em}

\begin{wrapfigure}[8]{l}{4cm}
   \resizebox{4cm}{!}{\includegraphics{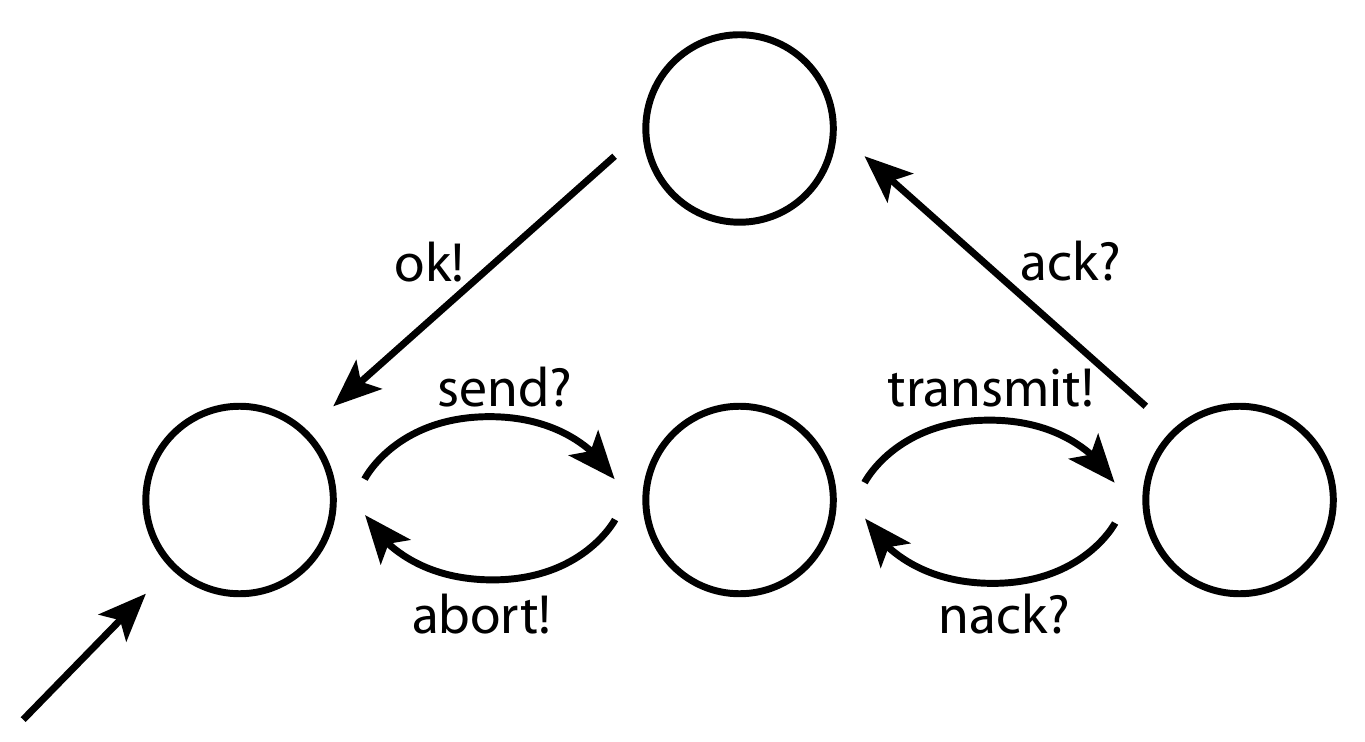}}
    \caption{\emph{BIA} \emph{Send}}
    \label{fig:spec}
\end{wrapfigure}
Consider a \emph{BIA} \emph{Send} in Figure~\ref{fig:spec}. It
receives a message via the input \emph{send?}. It then tries to send
this message over an unreliable medium using the output
\emph{transmit!}. In response to \emph{transmit?}, it can receive an
input {\em ack?} signifying successful delivery of the message, or an
input {\em nack?} signifying failure. It can then try to {\em
  transmit!} again (unboundedly many times), or it can abort using
the output {\em abort!}. {\em Send} will be our specification
interface. 


\begin{figure}[b]
  \begin{minipage}[b]{0.4\linewidth}
    \centering
    \resizebox{4cm}{!}{\includegraphics{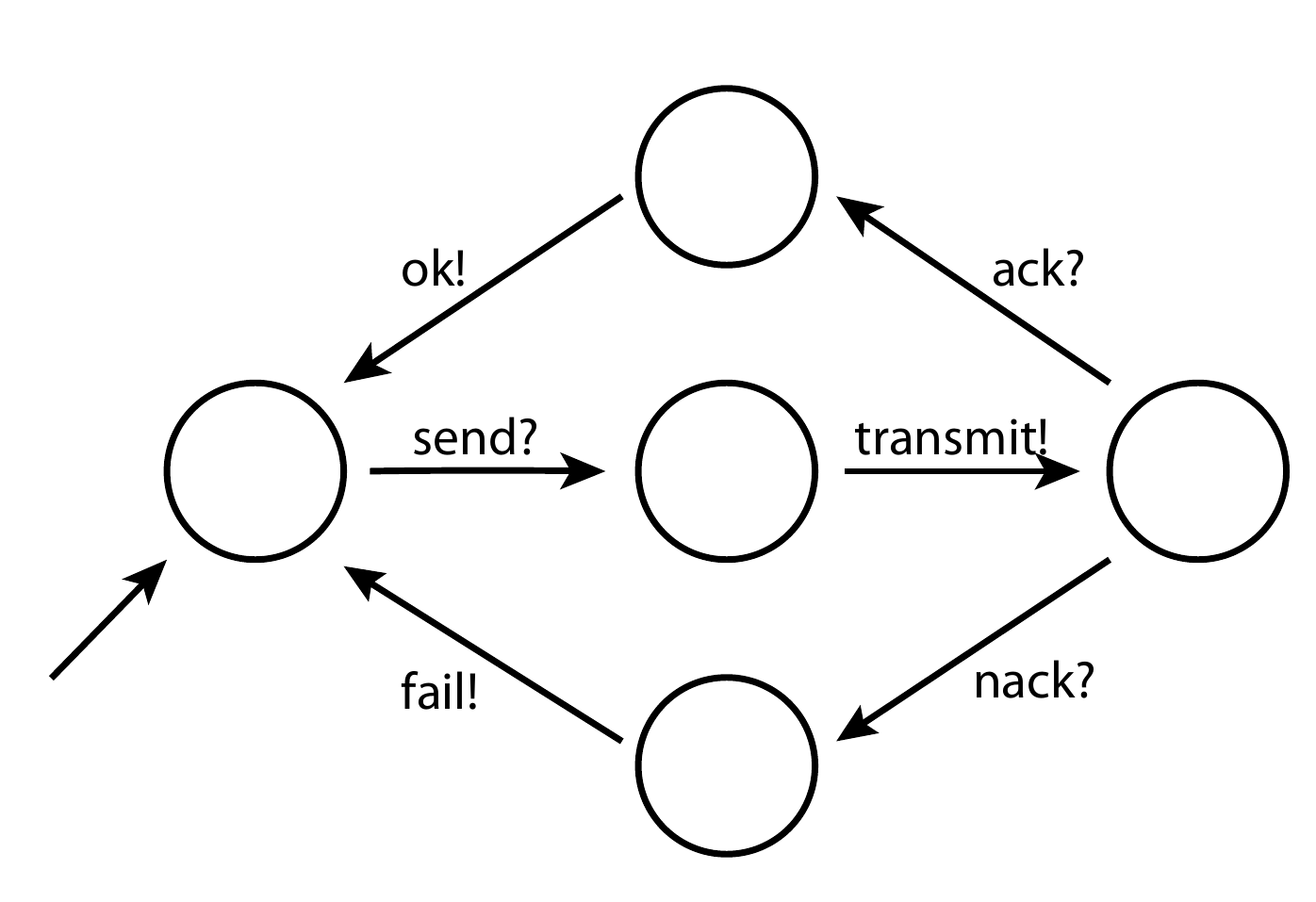}}
    \caption{Implementation \emph{SendOnce}}
    \label{fig:sendonce}
  \end{minipage}
  \begin{minipage}[b]{0.6\linewidth} 
   \centering
   \resizebox{7cm}{!}{\includegraphics{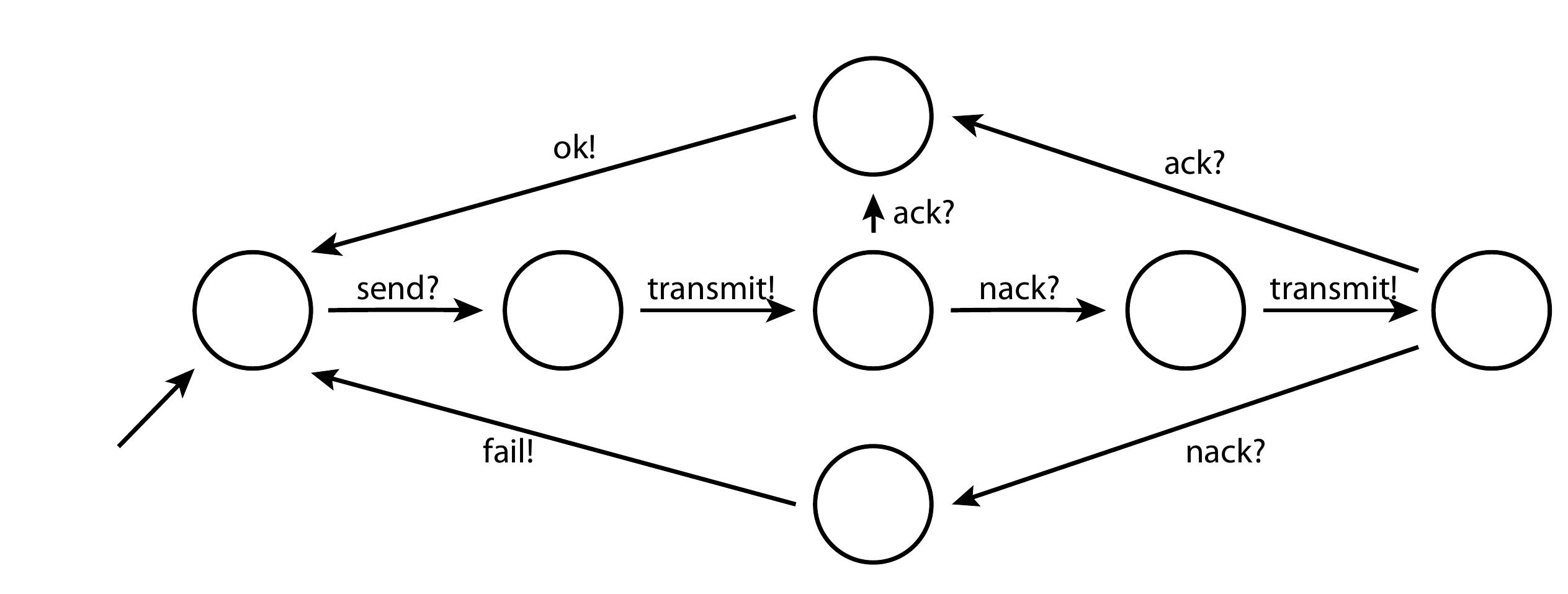}}
    \caption{Implementation \emph{SendTwice}}
    \label{fig:sendtwice}
  \end{minipage}
\end{figure}

We consider two implementation interfaces \emph{SendOnce} and
\emph{SendTwice} (Figures~\ref{fig:sendonce} and \ref{fig:sendtwice}).  
\emph{SendOnce} tries to send the message
only once and if it does not 
succeed it reports a failure by sending \emph{fail!} output. The second
implementation \emph{SendTwice} tries to send the message twice and
it reports a failure only if the transmission fails two times in a row. 
These implementation interfaces thus differ from the specification
interface which can try to transmit the message an unbounded number of
times. In particular, \emph{SendOnce} or 
\emph{SendTwice} do not refine the specification \emph{Send} in the
classical boolean sense.  

In order to compute distances between {\em Send} on one hand, and {\em
  SendOnce} and {\em SendTwice}, we first define an error model. The
output error model $\model_O$ we consider  
allows to play an output action \emph{fail!} instead of \emph{abort!}
with penalty $1$. We construct two games: $\game_{Send \otimes
  \model_O, SendOnce}$ and  
$\game_{Send \otimes \model_O, SendTwice}$.
The goal of Player~1 is to make Player~2 cheat by playing {\em abort!} 
as often as possible. Therefore, whenever
Player~1 has a choice between \emph{ack?} and \emph{nack?} the optimal
strategy is going to play \emph{nack?}. This agrees with the
intuition that the difference between {\em Send} and {\em SendOnce}
({\em SendTwice}) is manifested in the case when the transmission fails. 


The resulting distances are $d(Send \otimes \model_O, SendOnce) =
\frac{1}{4}$ and $d(Send \otimes \model_O, SendTwice) = \frac{1}{6}$.
According to the computed distances 
\emph{SendTwice} is closer to the specification than \emph{SendOnce},
as it tries to send the message before 
reporting a failure more times. 

In order to illustrate the behavior of the interface simulation
distance under composition of interfaces, we compose the interfaces
{\em Send}, {\em SendOnce}, and {\em SendTwice} with an interface
modeling the unreliable medium. The interface \emph{Medium} in
Figure~\ref{fig:wire} models an interface that 
fails to send a message at most two times in a row. The resulting
systems
\emph{Send $\parallel$ Medium}, \emph{SendOnce $\parallel$ Medium} and
\emph{SendTwice $\parallel$ Medium}  can be seen on
Figure~\ref{fig:send_wire}, \ref{fig:sendonce_wire} and
\ref{fig:sendtwice_wire}.
\begin{wrapfigure}[6]{l}{4.5cm}
   \centering
   \resizebox{4.5cm}{!}{\includegraphics{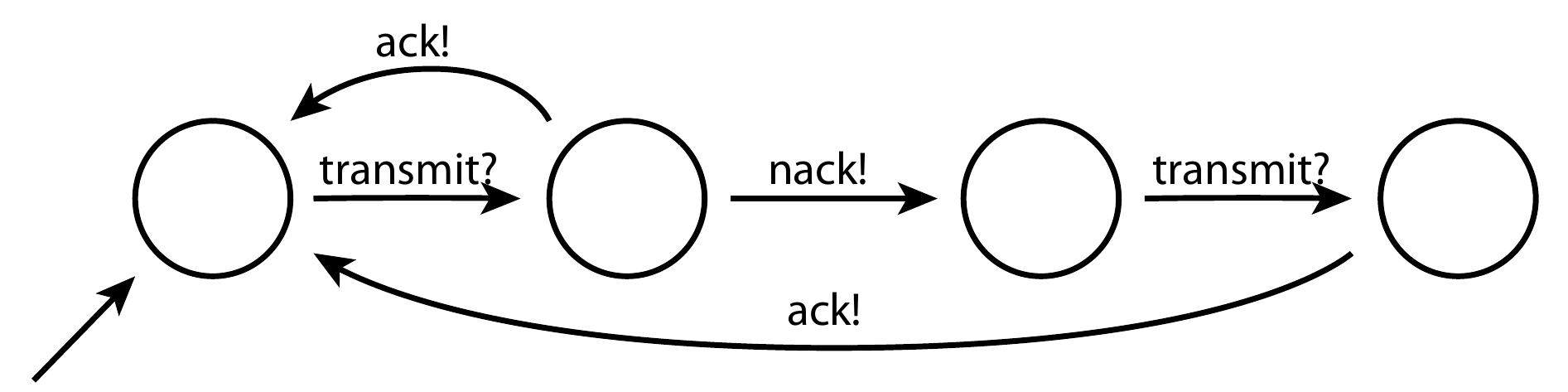}}
    \caption{The \emph{Medium}}
    \label{fig:wire}
\end{wrapfigure}
\noindent Again we can construct two games and compute their
values. We obtain:
$d((Send \parallel Medium) \otimes \model_O, SendOnce \parallel
    Medium) = \frac{1}{8}$, and 
$d((Send \parallel Medium) \otimes \model_O,
    SendTwice \parallel Medium) = 0$ . 
As expected, when the \emph{Medium} cannot fail two times in a row the
implementation \emph{SendTwice} is as good as the specification and
therefore the distance would be $0$.  We remark that if we would
change the model of the \emph{Medium} to the one that 
never fails, both the distances would be $0$.

\begin{figure}[h]
   \begin{minipage}[b]{0.6\linewidth}
   \centering
   \resizebox{7cm}{!}{\includegraphics{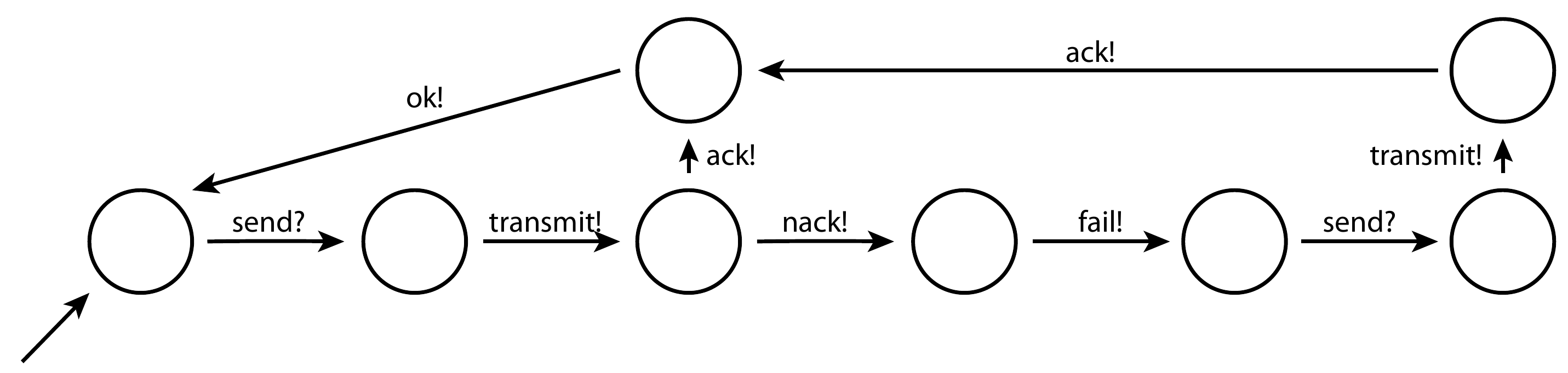}}
    \caption{The \emph{SendOnce $\parallel$ Medium}}
    \label{fig:sendonce_wire}
  \end{minipage}
   \begin{minipage}[b]{0.4\linewidth}
   \centering
   \resizebox{4.5cm}{!}{\includegraphics{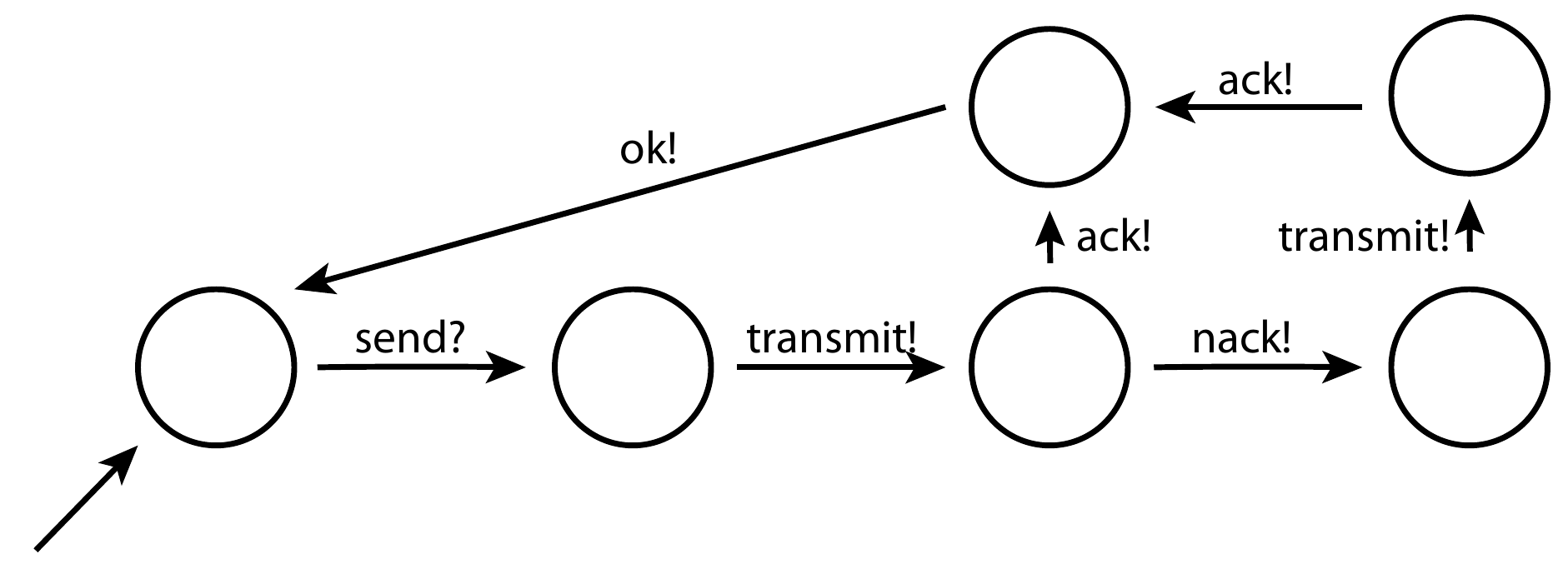}}
   \caption{\emph{SendTwice $\parallel$ Medium}}
   \label{fig:sendtwice_wire}
  \end{minipage}
\end{figure}

\vspace{-1.3em}
\subsection{Error Correcting Codes}
\vspace{-0.3em}
\begin{wrapfigure}[10]{}{6cm}
   \resizebox{6cm}{!}{\includegraphics{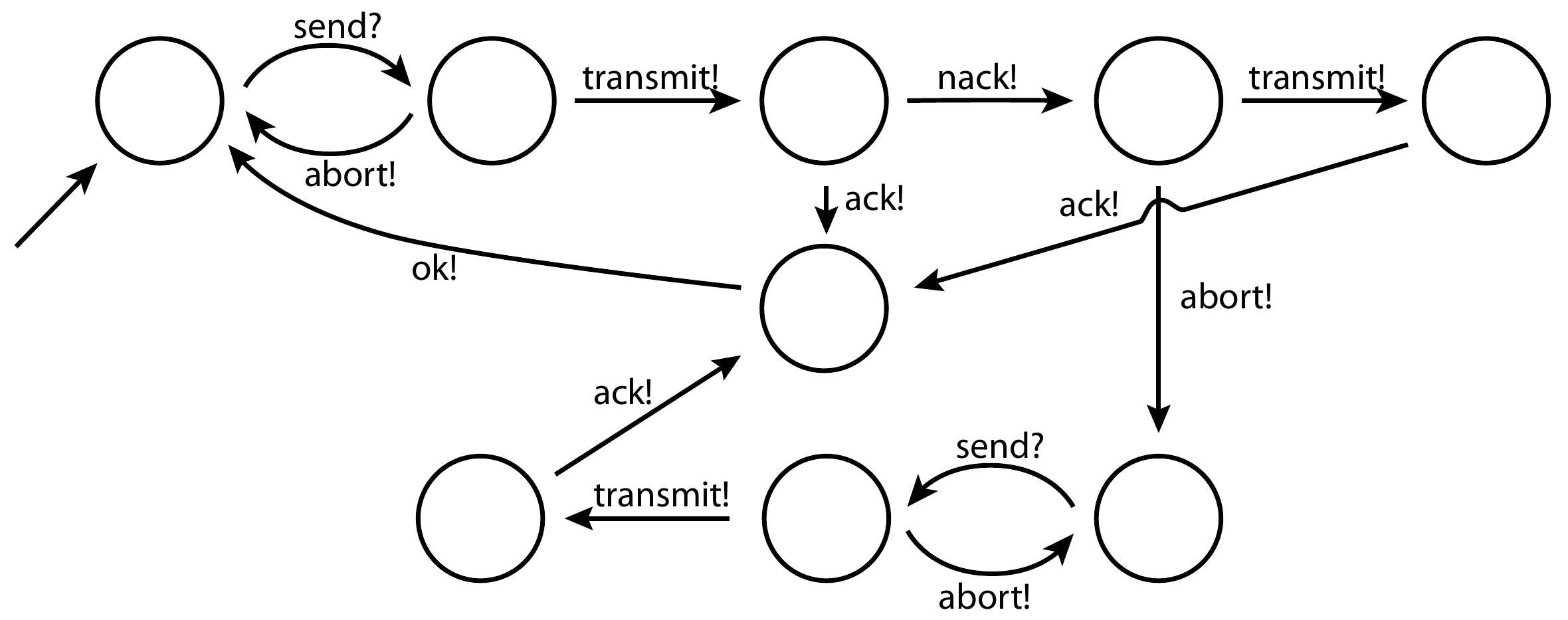}}
    \caption{\emph{Send $\parallel$ Medium}}
    \label{fig:send_wire}
\end{wrapfigure}
Error correcting codes are a way to ensure reliable information
transfer through a noisy environment. An error correcting code is
for our purposes a function that assigns every 
binary input string a fixed length {\em codeword} -- again a binary string ---
that is afterwards transmitted. 
A natural way how to improve
the chances of a correct transfer is to encode each message into a
codeword by adding redundant bits. These codewords might get
corrupted during the transmission, but the redundancy will cause that
codewords are not close to each other (according to Hamming distance),
and therefore it is possible to detect erroneous transfer, and 
sometimes even to correct some of the errors.  
Note that in what follows, we consider a situation where the only
type of error allowed during the transmission is a bit flip. 

We consider the well-known standard {\em $(n,M,d)$-code}, where $n$ is
the length of 
the code words, $M$ is the number of different original messages, and
$d$ is the minimal Hamming distance between codewords. 
For instance, if we are given an error correcting code such that the minimal
distance between codewords is 3 (i.e. an $(n,M,3)$-code for some $n$ and
$M$), then whenever a single bit flip 
occurs we receive a string that is not among the codewords. However,
there exists a unique codeword such that it has the minimal distance
to the received string. The received string can be then corrected to
this codeword.  

We consider two different error correcting codes
$C_1$ and $C_2$. Both codes encode
$2$ bit strings into $5$ bit codewords. The codes are given in the
following table: 

\begin{tabular*}{10in}{ccccc}
$C_1(00) = 00000 \:$ & $C_1(01) = 00101\:$ & \hspace{1cm} & $C_2(00) =
00000\:$ & $C_2(01) = 01101\:$ \\ 
$C_1(10) = 10110 \:$ & $C_1(11) = 11011\:$ &  & $C_2(10) = 10110\:$ &
$C_2(11) = 11011\:$ \\ 
\end{tabular*}

Note that $C_1$ is a $(5,4,2)$ code, i.e., its codewords have length
$5$, it encodes $4$ words and the minimal Hamming distance between two
codewords is $2$. The minimal distance $2$ ensures that when decoding the
codeword a single bit flip can be detected, however, not corrected. On
the other hand the code $C_2$ is a $(5,4,3)$ code and therefore can
detect $2$ bit flips and correct a single bit flip. 

\begin{wrapfigure}[20]{r}{5cm}
\vspace*{-1cm}
   \resizebox{5cm}{!}{\includegraphics{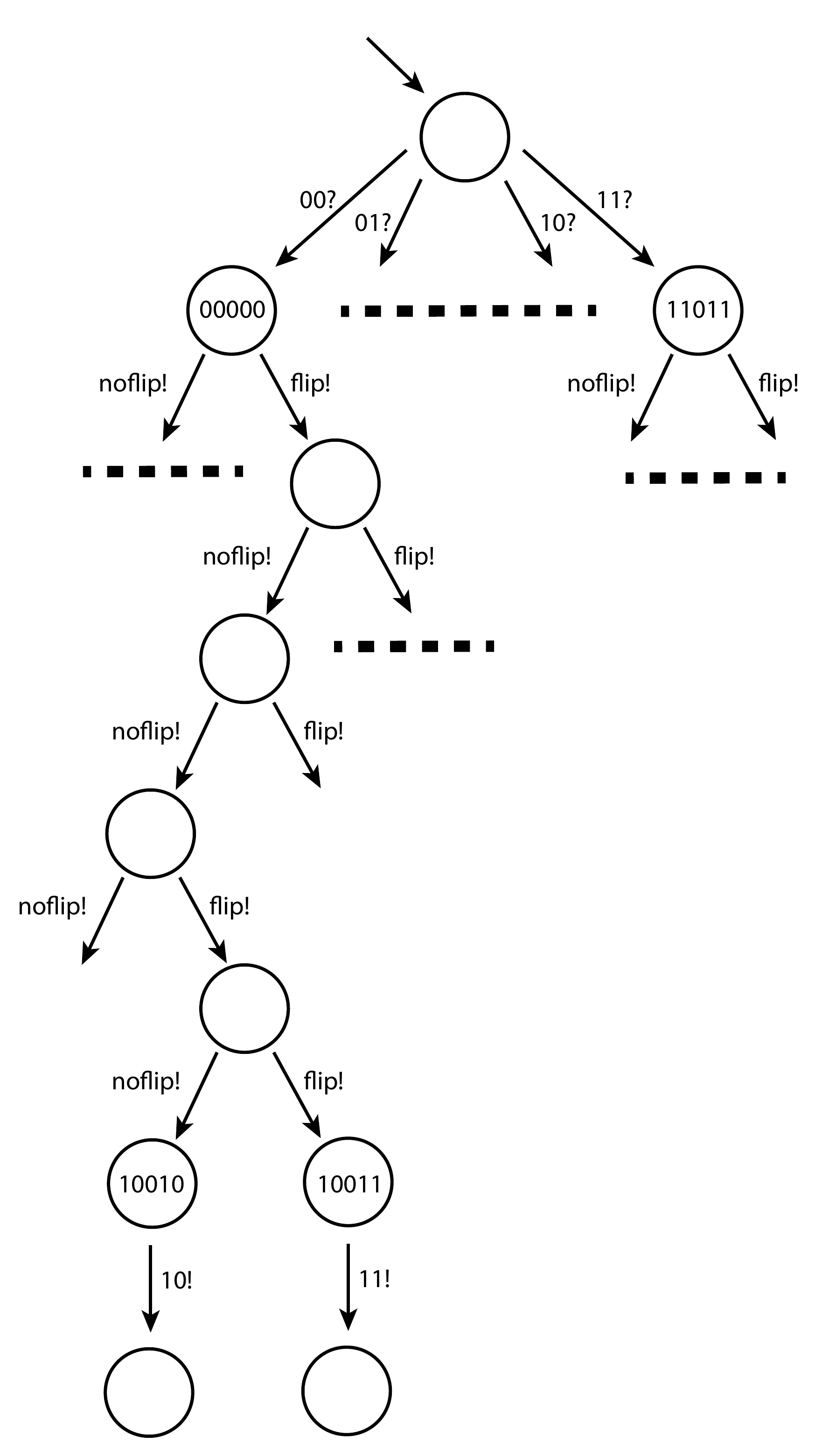}}
   \caption{Code $C_2$}
   \label{fig:error_code}
\vspace*{-0.5cm}
\end{wrapfigure}

We model as BIAs  the codes $C_1$ and $C_2$ and their transmission over
a network where bit flips can occur.  
We construct the \emph{BIAs} $F_{C_1}$ and $F_{C_2}$ 
according to the scheme presented in Figure~\ref{fig:error_code} (this
scheme is inside a loop and thus occurs repeatedly in both the BIAs). The
first action is the input of a two-bit word that should be
transmitted. The input word is then encoded according to $C_1$ (in
$F_{C_1}$), or $C_2$ (in $F_{C_2}$). 
Then a sequence of five actions \emph{flip} (or \emph{noflip})
determines whether a bit flip occurs on the corresponding
position. Depending on the \emph{flip}/\emph{noflip} sequence received
and the error correcting 
code used, the final output is the decoding of the received string,
with possibly some of the corrupted bits detected or repaired. More
precisely, on an input $x$, 
{$F_{C_1}$} can detect a single bit flip, and could in this case 
  send an \emph{error} output. In case of more
  flips, it can even output a symbol different from the input $x$.
Similarly, on an input $x$, {$F_{C_2}$}, in case of a single bit flip, can
  detect and correct the bit flip, and output the the message $x$. If there
  are multiple flips it outputs a string different from the input
  $x$. 
As a specification interface, we consider a BIA $F_{Spec}$ that uses
the schema from Figure~\ref{fig:error_code}, but always outputs its
input message, no matter what sequence of actions
\emph{flip} or \emph{noflip} it receives. 

We compose all three automata with a \emph{BIA} $F_{Error}$ modeling
the allowed number of bit flips. Let $F_{Error}$ allow only a single
bit flip in $5$ bits. The output error model $\model_O$ allows the
Player~2 to play all the output $2$ bit strings together with the
\emph{error} actions interchangeably. Then the corresponding values of
the games are as follows: 
(a) $d((F_{C_{Spec}} \parallel F_{Error}) \times \model_O,
  F_{C_{1}} \parallel F_{Error}) = 0 $, and (b)
  $d((F_{C_{Spec}} \parallel F_{Error}) \times \model_O,
  F_{C_{2}} \parallel F_{Error}) = \frac{1}{7} $.
This shows the that the code $C_2$ performs better than the code
$C_1$, as it can 
not only detect bit flips, but can also correct a single bit flip. In
case we would use a $F_{Error}$ that could do multiple bit flips in
$5$ bits, then distances of both codes would be the same.

\vspace{-1.0em}
\section{Conclusion}
\vspace{-0.3em}

\mypara{Summary.}
This paper extends the quantitative notion of simulation
distances~\cite{CHR10a} to automata with inputs and outputs. 
This distance relaxes the boolean notion of refinement and allows us to
measure the ``desirability'' of an interface with respect to a
specification, or select the best fitting interface from several
choices that do not refine a specification interface in the usual
boolean sense. We show that the interface simulation distance is a
directed metric, i.e., it satisfies reflexivity and the triangle
inequality. Moreover, the distance can only decrease when the
interface are composed with a third interface, which allows us to
decompose the specification into simpler parts. Furthermore, we show
that the distance can be bounded from above and below by considering
abstractions of the interfaces.  

\mypara{Future work.} 
Defining the interface simulation distance for broadcast
interface automata is one particular instance of a more general
idea. We plan to examine the properties of the
distance on other types of I/O automata, with
differing notions of composition, with internal actions, or timed
automata and automata modeling resource usage. Second, we plan
to investigate probabilistic versions of the simulation distances,
which would be 
useful in cases where there is a probability distribution on possible
environment inputs. Third, we plan to perform larger case studies to
establish which error models and accumulating functions (LimAvg,
$\Disc_{\lambda}$, etc.) are most useful in practice. 
\vspace{-0.5em}
\bibliographystyle{eptcs}
\bibliography{cite}

\end{document}